\documentclass{article}

\usepackage{amsmath,amssymb,amsthm}
\usepackage{fullpage}
\usepackage{microtype}
\usepackage{enumitem}
\usepackage{thmtools}
\usepackage{thm-restate}
\usepackage{xcolor}
\usepackage{algorithm}
\usepackage[
    indLines=true,
    noEnd=true,
    rightComments=true,
    italicComments=true,
]{algpseudocodex}
\algrenewcommand\algorithmicrequire{\textbf{Input:}}
\algrenewcommand\algorithmicensure{\textbf{Output:}}

\declaretheorem[name=Theorem]{theorem}
\declaretheorem[name=Lemma,sibling=theorem]{lemma}

\usepackage[
    colorlinks=true,
    linkcolor=blue!62!black,
    citecolor=green!48!black,
    urlcolor=blue!70!black,
    linktoc=page
]{hyperref}
\usepackage[nameinlink,noabbrev]{cleveref}

\crefname{theorem}{Theorem}{Theorems}
\Crefname{theorem}{Theorem}{Theorems}
\crefname{lemma}{Lemma}{Lemmas}
\Crefname{lemma}{Lemma}{Lemmas}
\crefname{algorithm}{Algorithm}{Algorithms}
\Crefname{algorithm}{Algorithm}{Algorithms}
\crefname{section}{Section}{Sections}
\Crefname{section}{Section}{Sections}

\newcommand{\algline}[2]{%
    \hyperref[#2]{Line~\ref*{#2}} of \cref{#1}%
}

\newcommand{\supp}{\operatorname{supp}}
\newcommand{\val}{\operatorname{val}}
\newcommand{\norm}[1]{\left\lVert #1\right\rVert}

\title{A Simple Las Vegas Algorithm for Sparse Nonnegative Convolution}
\author{
    Trevor Vaughn\thanks{Carnegie Mellon University. Email: tnvaughn@cmu.edu}
}

\begin{document}
\maketitle

\begin{abstract}
    Let $A, B \in \mathbb{Z}_{\ge 0}^n$ be nonnegative vectors and let $t = |\supp(A \star B)|$. We give a Las Vegas algorithm that computes $A \star B$ in $O(t \log t)$ expected time. More generally, for every $0 < \delta \le \frac{1}{2}$, the algorithm terminates within $O(t \log t \log \frac{1}{\delta})$ time with probability at least $1 - \delta$. The algorithm uses dense convolution, linear hashing, and the length reduction of \cite{BFN22}. Its main ingredient is a carry-free representation of the indices as vectors of constant dimension $d$ whose coordinates have size $O(t / \log t)$. We can then take our hash function to be the inner product with a random element of $\mathbb{F}_p^d$ for a prime $p$ of size $\Omega(t / \log t)$: this preserves addition and gives collision probability exactly $1/p$, while identities regarding the moments of the vectors identify and recover the isolated terms as in \cite{BFN22}. Our expected running time matches that of Jin and Xu~\cite{JX24} while using substantially different tools and yielding a simpler algorithm. Note that their algorithm also terminates within $O(t \log t)$ time with probability at least $1 - \frac{1}{t}$, while our tail bound is weaker.
\end{abstract}

\section{Introduction}
The convolution of two vectors $A, B \in \mathbb{Z}^n$ is the vector $C = A \star B$ defined by $C_z = \sum_{x + y = z} A_x B_y$. Dense convolution can be solved in $O(n \log n)$ time by the Fast Fourier Transform. When the inputs are sparse and the output has only $t = |\supp(A \star B)|$ nonzero positions, however, we would like a running time depending on $t$ rather than $n$, which is potentially much larger. The output-sensitive version of the problem was first solved in near-linear time by Cole and Hariharan~\cite{CH02}. Their Las Vegas algorithm runs in $O(t \log^2 n)$ time. Bringmann, Fischer, and Nakos subsequently gave an $O(t \log t + \log^{O(1)} n)$ Monte Carlo algorithm \cite{BFN21}. They also developed more direct Las Vegas algorithms, including simple $O(t \log^2 t)$-expected time and $O(t \log t \log \log t)$-expected time algorithms, together with the first deterministic near-linear time algorithm~\cite{BFN22}. More recently, Jin and Xu used the large sieve inequality to obtain an $O(t \log t)$-expected time Las Vegas algorithm which also runs in $O(t \log t)$ with probability $1 - \frac{1}{t}$ \cite{JX24}. We give a Las Vegas algorithm with $O(t \log t)$ expected running time. Our approach is closer to the hashing and recovery algorithms of \cite{BFN22} while the $O(t \log t)$-time algorithm of \cite{BFN21} and the $O(t \log t)$-time Las Vegas algorithm of \cite{JX24} use broader collections of tools and are more complicated. Our result is weaker than Jin and Xu as our algorithm can be modified to terminate in time $O(t \log^2 t)$ with probability $1 - \frac{1}{t}$ rather than $O(t \log t)$.

Our main result is the following.

\begin{restatable}[Sparse nonnegative convolution]{theorem}{mainthm}
\label{thm:main}
Let \(A,B\in\mathbb Z_{\ge0}^n\), and let
\(t=|\supp(A\star B)|\).  There is a Las Vegas algorithm that computes
\(A\star B\) in \(O(t\log(t+2))\) expected time.  Moreover, for every
\(0<\delta\le1/2\), the algorithm can be scheduled to terminate within
\[
    O\bigl(t\log(t+2)\log(1/\delta)\bigr)
\]
time with probability at least \(1-\delta\).
\end{restatable}

\paragraph{Technical overview.}
We first consider the case in which the length \(n\) is polynomial in a
guess \(T\ge t\) for the output sparsity.  Set
\(\beta=\Theta(T/\log T)\) and represent every index by its
\(d=\lceil\log_\beta n\rceil=O(1)\) base-\(\beta\) digits.  We then lift
\(A\) and \(B\) to \(d\)-dimensional vectors and replace ordinary addition
by coordinatewise addition.  This removes carries and makes the lifted
convolution compatible with the additive hash we use.  The lift does not
substantially increase the output sparsity: for each scalar output index,
the \(d\) carry bits determine the corresponding coordinatewise sum, so
that at most \(2^d=O(1)\) lifted indices map to it.  Consequently, the
lifted convolution has \(O(t)\) nonzero terms.

A standard additive hash maps an integer \(x\) to \(x\bmod p\), where
\(p\) is a random prime of some prescribed size \(P\).  This hash has two
limitations for our purposes.  First, two indices \(x\ne y\) collide
precisely when \(p\mid x-y\).  Since one difference may have several prime
divisors in the sampled range, the resulting collision probability is
only \(O(\log n/P)\), rather than \(O(1/P)\).

The carry-free lift avoids this problem.  Choose a fixed prime \(p>2\beta\),
sample \(\boldsymbol\alpha\in\mathbb F_p^d\) uniformly, and define
\(h(\mathbf u)=\langle\boldsymbol\alpha,\mathbf u\rangle\bmod p\).
Every coordinate difference between two lifted indices has absolute value
less than \(2\beta<p\), so distinct lifted indices remain distinct in
\(\mathbb F_p^d\).  The random linear functional therefore makes any fixed
pair collide with probability exactly \(1/p\).  The hash is additive because of the coordinatewise addition.
Hence the weights and moments of all hashed buckets can be computed from
the lifted inputs using a constant number of cyclic convolutions.  The
zeroth, first, and second moments determine whether a bucket contains a
single lifted term: after normalizing its weights to a probability
distribution, the relevant identity is precisely the condition that the
sum of the coordinate variances is zero.  We recover only such singleton
buckets, so every recovered term is correct.

Suppose at most \(r\) lifted terms remain.  With \(p=\Theta(r)\), a union
bound shows that any fixed term collides with another one with probability
bounded away from one.  We use geometrically decreasing sizes for the moduli
and repeat the recovery pass more often at later stages, reducing the
residual support to \(\beta=\Theta(T/\log T)\) with constant probability.
We then perform \(O(\log T)\) final passes with
\(p=\Theta(\beta)\) to recover every remaining term with constant probability.  The choice of
\(\beta\) yields the desired running time: a final convolution
costs \(O(\beta\log\beta)=O(T)\), so the final stage costs
\(O(T\log T)\).  The earlier stages form a geometrically decreasing sum
of the same order.  Finally, nonnegativity certifies completeness through checking against
\(\norm{A\star B}_1=\norm{A}_1\norm{B}_1\) since our output is always coordinatewise less than $A \star B$.

The carry-free lift requires \(n\le T^{O(1)}\).  We remove this assumption
using the length reduction of Bringmann, Fischer, and Nakos~\cite{BFN22},
which maps the original instance to a constant number of instances of
polynomial length and output sparsity \(O(t)\).  Trying geometrically
increasing values of \(T\) then gives the final Las Vegas algorithm.

\section{Preliminaries}
For $n \in \mathbb{N}$, we write $[n] = \{0, \dots, n-1\}$. For a vector $V$ we write $\supp(V) = \{i : V_i \ne 0\}$, and write $\norm{V}_1 = \sum_i |V_i|$. Vectors are given as lists of their nonzero coordinates and values at those points. For $A, B \in \mathbb{Z}^n_{\ge 0}$, define $C = A \star B \in \mathbb{Z}^{2n-1}_{\ge0}$ by $C_z = \sum_{\substack{x, y \in [n]\\x+y=z}} A_x B_y$. We write $a = |\supp A|$, $b = |\supp B|$, and $t = |\supp C| = |\supp (A \star B)|$.

If either top-level input is identically zero, we return the sparsely
represented zero vector immediately.  We henceforth assume that \(A\)
and \(B\) are not identically zero.  Let
\[
    \Delta:=\max\{\norm{A}_\infty,\norm{B}_\infty\}.
\]
For vectors \(P,Q\) indexed by \(\mathbb Z_m\), their cyclic convolution
is defined by
\[
    (P\star_m Q)_r
    =\sum_{\substack{x,y\in\mathbb Z_m\\x+y=r}}P_xQ_y.
\]

We work in the standard word-RAM model.  A word has
\[
    w=\Omega\!\left(
        \log(t+2)+\log\bigl(n(\Delta+1)\bigr)
    \right)
\]
bits, and arithmetic on integers occupying a constant number of words
takes constant time.  This is the usual convention for
sparse convolution; see, for example, \cite{BFN22,JX24}.

We use the standard fact that the ordinary or
cyclic convolution of two length-\(m\) integer vectors can be computed
in \(O(m\log(m+2))\) word operations by the Fast Fourier Transform whenever the input and
output coefficients occupy \(O(1)\) words.

Under the schedule in the proof of
\cref{thm:main}, every index and modulus used by the algorithm is
\(n^{O(1)}\), and every coefficient, moment, and product in a variance
test has magnitude at most
\(\bigl(n(\Delta+1)\bigr)^{O(1)}\).  Consequently all of these
quantities occupy \(O(1)\) words.  Unless an array is explicitly
materialized for a dense convolution, vectors of large length
are stored as sparse lists or dictionaries of their nonzero
index--value pairs.

\section{Carry-Free Lift}
\label{sec:carry_free}

We first analyze the case in which the length \(n\) is polynomial in the current output-sparsity estimate \(T\); \cref{sec:length_reduction} will reduce the general case to this setting.  The lift replaces each scalar index by a constant number of small coordinates.  Coordinatewise addition is then carry-free, making it compatible with the additive hash used in \cref{sec:isolated_terms}, while the limited number of possible carry patterns ensures that the output sparsity increases by only a constant factor.

\begin{lemma}
\label{lem:input_support}
We have \(a+b-1\le t\le ab\), and hence \(\max\{a,b\}\le t\).
\end{lemma}

\begin{proof}
The upper bound is immediate.  For the lower bound, write
\(\supp(A)=\{x_1<\cdots<x_a\}\) and
\(\supp(B)=\{y_1<\cdots<y_b\}\).  Nonnegativity implies that the
\(a+b-1\) distinct sums
\(x_1+y_1<\cdots<x_a+y_1<x_a+y_2<\cdots<x_a+y_b\) all belong to
\(\supp(C)\).
\end{proof}

Fix an integer \(T\ge2\), and set
\[
 \beta=\max\left\{2,
   \left\lceil\frac{T}{\lceil\log_2(T+2)\rceil}\right\rceil\right\},
 \qquad
 d=\max\{1,\lceil\log_\beta n\rceil\}.
\]
Represent \(x\in[n]\) by its base-\(\beta\) digit vector
\(\mathbf x\in\{0,\ldots,\beta-1\}^d\), and put
\(\val(\mathbf u)=\sum_{j=0}^{d-1}u_j\beta^j\).  Define
\(\widehat A_{\mathbf x}=A_x\) and \(\widehat B_{\mathbf x}=B_x\), and
let \(D=\widehat A\star\widehat B\) be their \(d\)-dimensional
convolution:
\[
 D_{\mathbf u}
 =\sum_{\mathbf x+\mathbf y=\mathbf u}
   \widehat A_{\mathbf x}\widehat B_{\mathbf y},
 \qquad
 \mathbf u\in\{0,\ldots,2\beta-2\}^d.
\]

\begin{lemma}[Carry-free lifting]
\label{lem:carry_free_support}
For every \(z\),
\[
 C_z=\sum_{\mathbf u:\,\val(\mathbf u)=z}D_{\mathbf u},
 \qquad
 |\supp(D)|\le2^dt.
\]
If \(n\le T^c\) for a constant \(c\), then \(d=O_c(1)\) and consequently
\(|\supp(D)|=O_c(t)\).
\end{lemma}

\begin{proof}
The first identity groups the pairs \(x+y=z\) according to their
digitwise sum.

Fix \(z\in\{0,\ldots,2n-2\}\).  Since \(n\le\beta^d\), we have
\(z<2\beta^d\), and hence \(z\) has a unique representation
\[
    z=\sum_{j=0}^{d-1}z_j\beta^j+z_d\beta^d,
    \qquad
    z_j\in\{0,\ldots,\beta-1\},\quad z_d\in\{0,1\}.
\]
For any
\(\mathbf u\in\{0,\ldots,2\beta-2\}^d\) satisfying
\(\val(\mathbf u)=z\), let \(\kappa_j\) be the carry into coordinate
\(j\), including the terminal carry \(\kappa_d\).  Then
\[
    \kappa_0=0,\qquad \kappa_d=z_d,
\]
and
\[
    u_j+\kappa_j=z_j+\beta\kappa_{j+1}
    \qquad(0\le j<d).
\]
Every \(\kappa_j\) belongs to \(\{0,1\}\).  Since the endpoint carries
\(\kappa_0\) and \(\kappa_d\) are fixed, the internal carry vector
\((\kappa_1,\ldots,\kappa_{d-1})\) determines \(\mathbf u\) through
\[
    u_j=z_j+\beta\kappa_{j+1}-\kappa_j.
\]
There are at most \(2^{d-1}\le2^d\) internal carry vectors, so at most
\(2^d\) lifted indices map to any fixed \(z\).

\(D_{\mathbf u}>0\) implies
\(C_{\val(\mathbf u)}>0\) by nonnegativity.  Summing the preceding bound
over the \(t\) indices in \(\supp(C)\) gives
\(|\supp(D)|\le2^dt\).

Finally, \(\beta=\Theta(T/\log(T+2))\), so
\(\log\beta=\Omega(\log T)\) outside a bounded range.
Consequently,
\(d\le1+\log n/\log\beta=O_c(1)\) whenever \(n\le T^c\); the bounded
values of \(T\) are absorbed by the constant.
\end{proof}

\section{Recovering Isolated Terms}
\label{sec:isolated_terms}

We recover a lifted term whenever it is isolated by an additive hash.
The benefit of the carry-free representation is that it gives
a hash with a pairwise collision probability of exactly \(1/p\) using a fixed prime.
Hashing the original indices by \(x\mapsto x\bmod p\) cannot
separate two indices congruent modulo \(p\), and multiplying by a random
nonzero scalar does not change this.  Sampling the prime instead incurs
a logarithmic loss caused by the possible prime divisors of
the difference.  In the lifted representation, by contrast, all
coordinate differences are smaller than the modulus used.

Let \(p>2\beta\) be prime, choose
\(\boldsymbol\alpha\in\mathbb F_p^d\) uniformly, and define
\(h(\mathbf u)=\langle\boldsymbol\alpha,\mathbf u\rangle\bmod p\).
The hash is additive:
\(h(\mathbf u+\mathbf v)=h(\mathbf u)+h(\mathbf v)\).

If
\(\mathbf u\ne\mathbf v\in\{0,\ldots,2\beta-2\}^d\), then
\[
 \Pr[h(\mathbf u)=h(\mathbf v)]=\frac1p.
\]
To see this, choose a coordinate \(j\) for which \(u_j\ne v_j\).
Since \(|u_j-v_j|<2\beta<p\), this difference is nonzero in
\(\mathbb F_p\).  After fixing all coordinates of
\(\boldsymbol\alpha\) except \(\alpha_j\), exactly one value of
\(\alpha_j\) satisfies
\(\langle\boldsymbol\alpha,\mathbf u-\mathbf v\rangle=0\).

Let \(R\) be any nonnegative vector on
\(\{0,\ldots,2\beta-2\}^d\). Define the bucket at $r$ to be $\{\mathbf u : h(\mathbf u) = r\}$. For each \(z\in\mathbb F_p\), define
\[
 X_z=\sum_{h(\mathbf u)=z}R_{\mathbf u},\qquad
 Y_{j,z}=\sum_{h(\mathbf u)=z}u_jR_{\mathbf u},\qquad
 Z_z=\sum_{h(\mathbf u)=z}\norm{\mathbf u}_2^2R_{\mathbf u}.
\]

These moments then determine whether only a single index was hashed to a bucket, using the fact that a discrete distribution is constant if and only if it has zero variance.

\begin{lemma}[Singleton test]
\label{lem:singleton_test}
A nonempty bucket at \(z\) contains exactly one point of \(\supp(R)\) if
and only if $X_z \ne 0$ and
\[
 X_zZ_z=\sum_{j=0}^{d-1}Y_{j,z}^2.
\]
In that case its point and coefficient are
\(u_j=Y_{j,z}/X_z\) and \(R_{\mathbf u}=X_z\).
\end{lemma}

\begin{proof}
Choose a random vector \(\mathbf U\) from the bucket according to
\(\Pr[\mathbf U=\mathbf u]=R_{\mathbf u}/X_z\).  For every coordinate
\(j\),
\[
 \mathbb E[U_j]
 =\frac{Y_{j,z}}{X_z},
 \qquad
 \mathbb E[U_j^2]
 =\frac{1}{X_z}
   \sum_{h(\mathbf u)=z}u_j^2R_{\mathbf u}.
\]
Therefore
\[
\begin{aligned}
 X_z^2\sum_{j=0}^{d-1}\operatorname{Var}(U_j)
 &=
 X_z^2\sum_{j=0}^{d-1}
 \left(\mathbb E[U_j^2]-\mathbb E[U_j]^2\right)
 =
 X_z\sum_{j=0}^{d-1}
 \sum_{h(\mathbf u)=z}u_j^2R_{\mathbf u}
 -\sum_{j=0}^{d-1}Y_{j,z}^2\\
 &=
 X_z\sum_{h(\mathbf u)=z}
 \norm{\mathbf u}_2^2R_{\mathbf u}
 -\sum_{j=0}^{d-1}Y_{j,z}^2
 =X_zZ_z-\sum_{j=0}^{d-1}Y_{j,z}^2.
\end{aligned}
\]
Thus it vanishes if and only if every coordinate \(U_j\)
has variance zero.  A discrete random variable has variance zero exactly
when it is constant, so this occurs if and only if \(\mathbf U\) is
supported at a single vector.  Equivalently, the bucket contains exactly
one point of \(\supp(R)\).  In that case
\(Y_{j,z}=u_jX_z\), and hence
\(u_j=Y_{j,z}/X_z\) and \(R_{\mathbf u}=X_z\).
\end{proof}

The required values can be computed without forming \(D\) using additivity.  For
\(r\in\mathbb F_p\), define
\[
\begin{array}{lll}
 A^{(0)}_r=\displaystyle\sum_{h(\mathbf x)=r}\widehat A_{\mathbf x},
&
 A^{(j)}_r=\displaystyle\sum_{h(\mathbf x)=r}
                   x_j\widehat A_{\mathbf x},
&
 A^{(\ell_2)}_r=\displaystyle\sum_{h(\mathbf x)=r}
          \norm{\mathbf x}_2^2\widehat A_{\mathbf x},
\end{array}
\]
and define the corresponding vectors for \(\widehat B\).  Write
\(\star_p\) for cyclic convolution of length \(p\).

\begin{lemma}[Computing the bucket statistics]
\label{lem:moment_convolution}
For \(R=D\), the statistics in \cref{lem:singleton_test} satisfy
\[
\begin{split}
 X&=A^{(0)}\star_pB^{(0)},\\
 Y_j&=A^{(j)}\star_pB^{(0)}
      +A^{(0)}\star_pB^{(j)},\\
 Z&=A^{(\ell_2)}\star_pB^{(0)}
      +A^{(0)}\star_pB^{(\ell_2)}
      +2\sum_{j=0}^{d-1}A^{(j)}\star_pB^{(j)}.
\end{split}
\]
They can be computed in
\(O(d(a+b)+dp\log p)\) operations.  Given already recovered terms
\((\mathbf u,D_{\mathbf u})\), their contributions can be subtracted in
\(O(d)\) time per term.
\end{lemma}

\begin{proof}
Every pair \((\mathbf x,\mathbf y)\) contributes to bucket
\(h(\mathbf x+\mathbf y)=h(\mathbf x)+h(\mathbf y)\).  The formulas
then follow from \(u_j=x_j+y_j\) and
\(\norm{\mathbf x+\mathbf y}_2^2
=\norm{\mathbf x}_2^2+\norm{\mathbf y}_2^2
 +2\sum_jx_jy_j\).  They use \(3d+3\) cyclic convolutions.  A recovered
term of weight \(w\) is removed by subtracting \(w\), \(u_jw\), and
\(\norm{\mathbf u}_2^2w\) from its bucket.
\end{proof}

\begin{algorithm}[H]
\caption{\(\textsc{RecoveryPass}(\widehat A,\widehat B,\beta,p,
                                  \boldsymbol\alpha,\mathcal R)\)}
\label{alg:recovery_pass}
\begin{algorithmic}[1]
\Require Lifted inputs and a dictionary \(\mathcal R\) of recovered
terms
\Ensure Additional terms of \(D\)
\State Compute the statistics of \(D-\mathcal R\) using
       \cref{lem:moment_convolution}
\State \(\mathcal N\gets\emptyset\)
\For{\(z\in\mathbb F_p\)}
    \If{\(X_z>0\) and \(X_zZ_z=\sum_jY_{j,z}^2\)}
        \State \(u_j\gets Y_{j,z}/X_z\) for every \(j\)
        \If{all \(u_j\) are integral,
             \(\mathbf u\in\{0,\ldots,2\beta-2\}^d\), and
             \(h(\mathbf u)=z\)}
            \State Insert \((\mathbf u,X_z)\) into \(\mathcal N\)
        \EndIf
    \EndIf
\EndFor
\State \Return \(\mathcal N\)
\end{algorithmic}
\end{algorithm}

\begin{lemma}[Certification]
\label{lem:mass_certificate}
Every term returned by \cref{alg:recovery_pass} is a correct,
previously unrecovered term of \(D\).  A dictionary \(\mathcal R\) of
recovered terms is complete if and only if
\[
 \sum_{(\mathbf u,w)\in\mathcal R}w=\norm{A}_1\norm{B}_1.
\]
\end{lemma}

\begin{proof}
After the previously recovered terms are subtracted, the remaining
vector is nonnegative.  Soundness therefore follows inductively from
\cref{lem:singleton_test}.  Since
\(\norm{D}_1=\norm{\widehat A}_1\norm{\widehat B}_1
=\norm{A}_1\norm{B}_1\), the displayed equality
holds exactly when no positive term remains.
\end{proof}

\section{A Recovery Trial}
\label{sec:recovery_trial}

When at most \(r\) terms remain and \(p\ge64r\), any fixed term collides
with another residual term with probability at most \(1/64\).  A recovery
pass therefore removes most residual terms in expectation.  We use
moduli geometrically decreasing in size and perform progressively more
passes at later stages, making the probabilities that any stage fails to
halve the residual small enough for a union bound.

The geometric stages stop once at most
\(\beta=\Theta(T/\log T)\) terms remain with constant probability.  We then perform
\(O(\log T)\) passes with a modulus of size \(\Theta(\beta)\), which recovers all
remaining terms with constant probability.  Each final pass takes
\(O(T+\beta\log\beta)=O(T)\) time.  Thus the final stage
costs \(O(T\log T)\).
Note that the condition $p > 2 \beta$ holds throughout since all the upper bounds on the residual we assume are at least $\beta$ and we choose sufficient constant factors.

For a scale \(T\ge\max\{a,b,2\}\), construct the lift from
\cref{sec:carry_free} and set
\[
 r_0=2^dT,\qquad r_i=2^{-i}r_0,\qquad
 L=\left\lceil\log_2\frac{r_0}{\beta}\right\rceil.
\]
For \(i<L\), choose a prime
\(p_i\in[64\lceil r_i\rceil,128\lceil r_i\rceil]\), and choose
\(p_\star\in[64\beta,128\beta]\).  These primes exist by Bertrand's
postulate and can all be found by a linear sieve up to \(128 r_0\) in $O(r_0)$ time.

\begin{algorithm}[H]
\caption{\(\textsc{LiftedTrial}(A,B,T)\)}
\label{alg:lifted_trial}
\begin{algorithmic}[1]
\Require Nonnegative vectors \(A,B\) and
        \(T\ge\max\{|\supp(A)|,|\supp(B)|,2\}\)
\Ensure \(A\star B\), certified correct, or \(\mathsf{fail}\)
\If{\(\supp(A)=\emptyset\) or \(\supp(B)=\emptyset\)}
    \State \Return the sparsely represented zero convolution
\EndIf
\State Construct \(\beta,d,r_0,\ldots,r_L,p_0,\ldots,p_{L-1},p_\star\)
\State \(\mathcal R\gets\emptyset\)
\For{\(i=0,\ldots,L-1\)}
    \label{line:geometric-stage}
    \Comment{Stage \(i\)}
    \For{\(\ell=1,\ldots,i+1\)}
        \State Sample \(\boldsymbol\alpha\in\mathbb F_{p_i}^d\)
               uniformly
        \State Add to \(\mathcal R\) the terms returned by
               \(\Call{RecoveryPass}
               {\widehat A,\widehat B,\beta,p_i,
                \boldsymbol\alpha,\mathcal R}\)
        \If{\(|\mathcal R|>r_0\)}
            \State \Return \(\mathsf{fail}\)
        \EndIf
    \EndFor
\EndFor
\For{\(\ell=1,\ldots,
       \lceil3\log_{64}(T+2)\rceil\)}
    \label{line:final-stage}
    \Comment{Final stage}
    \State Sample \(\boldsymbol\alpha\in\mathbb F_{p_\star}^d\)
           uniformly
    \State Add to \(\mathcal R\) the terms returned by
           \(\Call{RecoveryPass}
           {\widehat A,\widehat B,\beta,p_\star,
            \boldsymbol\alpha,\mathcal R}\)
    \If{\(|\mathcal R|>r_0\)}
        \State \Return \(\mathsf{fail}\)
    \EndIf
\EndFor
\If{\(\sum_{(\mathbf u,w)\in\mathcal R}w
      \ne\norm{A}_1\norm{B}_1\)}
    \State \Return \(\mathsf{fail}\)
\EndIf
\State \(C\gets0\)
\ForAll{\((\mathbf u,w)\in\mathcal R\)}
    \State \(C_{\val(\mathbf u)}\gets C_{\val(\mathbf u)}+w\)
\EndFor
\State \Return \(C\)
\end{algorithmic}
\end{algorithm}

A stage consists of \(k\) consecutive recovery passes using the
same prime \(p\).  Thus each iteration \(i\) of the outer loop beginning
on \algline{alg:lifted_trial}{line:geometric-stage} is a stage consisting
of \(i+1\) passes with prime \(p_i\).  The loop beginning on
\algline{alg:lifted_trial}{line:final-stage} is one final stage consisting
of \(\lceil3\log_{64}(T+2)\rceil\) passes with prime \(p_\star\).

\begin{lemma}[Progress in one stage]
\label{lem:stage_progress}
Consider a stage that begins with at most \(r\) unrecovered terms and
consists of \(k\) independent recovery passes using a prime
\(p\ge64r\).  The probability that more than \(r/2\) terms remain after
the stage is at most \(2\cdot64^{-k}\).
\end{lemma}

\begin{proof}
Fix a term \(\mathbf u\) present at the beginning of the stage.
Condition on \(\mathbf u\) surviving the first \(j\) passes.  Before
pass \(j+1\), the residual contains at most \(r\) terms, and the fresh
choice of \(\boldsymbol\alpha\) causes \(\mathbf u\) to collide with
another residual term with probability at most
\((r-1)/p\le1/64\).  Therefore,
\(\mathbf u\) survives all \(k\) passes with probability at most
\(64^{-k}\).

The expected number of surviving terms is consequently at most
\(r64^{-k}\).  The conclusion follows by Markov's inequality.
\end{proof}

\begin{lemma}[Recovery trial]
\label{lem:lifted_trial}
Fix a constant \(c\), and suppose \(n\le T^c\) and
\(T \ge\max\{a,b,2\}\).  Then
\cref{alg:lifted_trial} takes \(O_c(T\log(T+2))\) operations.  It never returns an
incorrect answer.  If \(T\ge t\), it returns \(A\star B\) with
probability greater than \(9/10\).
\end{lemma}

\begin{proof}
If either input is zero, the early return is correct and satisfies all
claimed bounds.  We may therefore assume that both inputs are nonzero.
Soundness then follows from \cref{lem:mass_certificate} and
\cref{lem:carry_free_support}.
If \(T\ge t\), initially
\(|\supp(D)|\le r_0\).  Applying \cref{lem:stage_progress} with
\(r=r_i\) and \(k=i+1\), the probability that any of the first \(L\)
stages fails to halve the residual is at most
\[
 \sum_{i\ge0}2\cdot64^{-(i+1)}=\frac2{63}.
\]
If all these stages succeed, at most \(r_L\le\beta\) terms remain.
A fixed remaining term survives all final passes with probability at
most
\[
 64^{-\lceil3\log_{64}(T+2)\rceil}\le(T+2)^{-3}.
\]
A union bound over at most \(\beta\le T+2\) terms bounds final failure
by \((T+2)^{-2}\).  Since \(T\ge2\), the total failure probability is
less than \(1/10\).

By \cref{lem:carry_free_support}, \(d=O_c(1)\) and
\(r_0=O_c(T)\).  Moreover,
\(\beta=\Theta(T/\log(T+2))\), and hence
\(L=\lceil\log_2(r_0/\beta)\rceil=O_c(\log\log(T+2))\).

A recovery pass with prime \(p\) takes
\(O_c(T+p\log p)\) time since the two input supports have total
size \(O(T)\), and every pass begins with at most \(r_0=O_c(T)\)
recovered terms because of the cap on \(\mathcal R\).  The
remaining work consists of a constant number of length-\(p\) cyclic
convolutions.

Stage \(i\) performs \(i+1\) passes with \(p_i=\Theta(r_i)\), where
\(r_i=2^{-i}r_0\).  The total cost of the first \(L\) stages is therefore
\[
 O_c\left(
   T\sum_{i=0}^{L-1}(i+1)
   +\sum_{i=0}^{L-1}(i+1)r_i\log(r_i+2)
 \right)
 =O_c(T\log(T+2)).
\]
Here the first sum is \(O(TL^2)\), while the second is
\(O_c(T\log(T+2))\) by the geometric decrease of \(r_i\).

The final stage performs \(O(\log(T+2))\) passes with
\(p_\star=\Theta(\beta)\).  Since
\(p_\star\log p_\star=O(T)\), this stage also takes
\(O_c(T\log(T+2))\) time.

The sieve takes \(O_c(T)\) time.  Across the entire trial there are
\(O_c(T)\) recovered records.  Dictionary insertions and the final
aggregation can be implemented deterministically with balanced search
trees, or by sorting the records by \(\val(\mathbf u)\) and summing
equal keys, in \(O_c(T\log(T+2))\) time.  The resulting vector \(C\) is
stored sparsely.  This cost is within the claimed bound.  Finally, the
recovered terms and all arrays maintained during one pass occupy
\(O_c(T)\) words.
\end{proof}

\section{Length Reduction}
\label{sec:length_reduction}

The carry-free lift has constant dimension only when the length $n$ is at most polynomial in the sparsity guess $T$. We now reduce the general problem to this special case using the almost-additive hashing construction and length reduction methods of Bringmann, Fischer, and Nakos~\cite{BFN22}. Their hash is similar to reducing modulo a random prime but avoids the $\log^{O(1)} n$ cost for sampling such a prime at the cost of not being additive but rather satisfying a weaker but sufficient condition. We emulate their reduction; they hash the original universe to one of size polynomial in the input sparsity using $a$ and $b$ to form a lower bound. We compute the first and second moments of the hashed vectors using \cref{alg:convolution_trial} to determine which buckets contain contributions to a single original output coordinate, and use the $\ell_1$ norm of the output to certify that all coordinates have been recovered as before.

\begin{lemma}[Linear hashing without primes {\cite[Lemma~13]{BFN22}}]
\label{lem:bfn_hash}
Let \(N\ge m\).  There exist a family
\(\mathcal H_{N,m}\) of hash functions
\(h\colon[N]\to\mathbb Z_m\) and a set
\(\Phi_{N,m}\subseteq\mathbb Z_m\), fixed independently of the sampled
function \(h\), with the following properties.
\begin{enumerate}[label=(\roman*)]
    \item Sampling \(h\in\mathcal H_{N,m}\) and evaluating \(h(x)\)
    take constant time.
    \item There is an absolute constant \(\kappa\) such that, for all
    distinct \(x,y\in[N]\) and every \(q\in\mathbb Z_m\),
    \[
        \Pr_{h\in\mathcal H_{N,m}}
        [h(x)-h(y)=q]\le\frac{\kappa}{m}.
    \]
    \item The set \(\Phi_{N,m}\) has absolute-constant size, and every
    \(h\in\mathcal H_{N,m}\) satisfies the following: for all
    \(x,y\in[N]\) such that \(x+y\in[N]\), there exists
    \(\phi\in\Phi_{N,m}\) for which
    \[
        h(x)+h(y)=h(x+y)+\phi
        \qquad\text{in }\mathbb Z_m.
    \]
\end{enumerate}
\end{lemma}

Fix an absolute constant \(\varphi_0\) such that
\[
    |\Phi_{N,m}|\le\varphi_0
\]
for every \(N\ge m\), and set
\[
    m:=\left\lceil100\kappa\varphi_0^2(ab)^2\right\rceil.
\]
If \(2n\le m\), then \(n=O((ab)^2)\), so the original length is already
polynomial in the input sparsities and no length reduction is necessary.
Suppose henceforth that \(2n>m\), sample
\(h\in\mathcal H_{2n,m}\), and write
\[
    \Phi:=\Phi_{2n,m}.
\]
This set is fixed independently of the sampled \(h\).  Moreover,
\(h(x)\), \(h(y)\), and \(h(x+y)\) are defined for every
\(x,y\in[n]\), since \(x+y\in[2n]\).

For a vector \(V\) supported on \([2n]\), define
\((\partial V)_x=xV_x\), \((\partial^2V)_x=x^2V_x\), and
\((hV)_r=\sum_{x:\,h(x)=r}V_x\).  Let \(\star_m\) denote cyclic
convolution of length \(m\).

\begin{lemma}[Sparsity of the reduced convolutions]
\label{lem:reduced_sparsity}
For
\(P\in\{A,\partial A,\partial^2A\}\) and
\(Q\in\{B,\partial B,\partial^2B\}\), we have
\[
    |\supp(hP\star_mhQ)|\le|\Phi|t.
\]
Moreover, the ordinary convolution \(hP\star hQ\) has at most
\(2|\Phi|t\) nonzero coordinates.
\end{lemma}

\begin{proof}
If \(P_xQ_y>0\), then \(x\in\supp(A)\) and \(y\in\supp(B)\).  Hence
\(z:=x+y\) belongs to \(\supp(A\star B)\), and almost-additivity gives
\(h(x)+h(y)\in \{h(z)\}+\Phi\).  Therefore
\[
    \supp(hP\star_mhQ)
    \subseteq
    \bigcup_{z\in\supp(A\star B)}(\{h(z)\}+\Phi),
\]
which proves the first claim.

Let \(F=hP\star hQ\).  Since \(F\) is nonnegative, every nonzero
coordinate of \(F\) remains nonzero after reducing modulo \(m\).
Each coordinate of the cyclic convolution has at most two preimages in
the ordinary convolution, and hence
\(|\supp(F)|\le2|\supp(hP\star_mhQ)|\le2|\Phi|t\).
\end{proof}

We next use the scalar analog of \cref{lem:singleton_test}.  Compute
\[
\begin{split}
    X&=hA\star_mhB,\\
    Y&=h(\partial A)\star_mhB+hA\star_mh(\partial B),\\
    Z&=h(\partial^2A)\star_mhB
       +2h(\partial A)\star_mh(\partial B)
       +hA\star_mh(\partial^2B).
\end{split}
\]
For \(r\in\mathbb Z_m\), define
\[
    W^{(r)}_z
    :=\sum_{\substack{x+y=z\\h(x)+h(y)=r}}A_xB_y.
\]
Note that the condition $h(x) + h(y) = r$ is different from before, as we no longer have additivity, we also change the definition of bucket at $w$ accordingly to $\{(x, y) \in \supp(A) \times \supp(B) : h(x) + h(y) = w\}$. It's easy to see that \(X_r=\sum_zW^{(r)}_z\),
\(Y_r=\sum_zzW^{(r)}_z\), and
\(Z_r=\sum_zz^2W^{(r)}_z\).  Then $X_r > 0$ and $Y_r^2=X_rZ_r$ if and only if there is only one index $z$ contributing to this bucket.

\begin{algorithm}[H]
\caption{\(\textsc{ConvolutionTrial}(A,B,T)\)}
\label{alg:convolution_trial}
\begin{algorithmic}[1]
\Require Nonnegative \(A,B \in \mathbb Z_{\ge0}^n\) and
        \(T\ge\max\{|\supp(A)|,|\supp(B)|,2\}\)
\Ensure \(A\star B\), certified correct, or \(\mathsf{fail}\)
\State \(a\gets|\supp(A)|\), \(b\gets|\supp(B)|\)
\State \(m\gets\lceil100\kappa\varphi_0^2(ab)^2\rceil\)
\If{\(2n\le m\)}
    \State \Return \(\Call{LiftedTrial}{A,B,T}\)
\EndIf
\State Sample \(h\colon[2n]\to\mathbb Z_m\) using
       \cref{lem:bfn_hash}
\State Construct \(hA,h(\partial A),h(\partial^2A)\) and the analogous
       vectors for \(B\)
\State \(\mathcal Q\gets\{(hA,hB),(h(\partial A),hB),
       (hA,h(\partial B)),(h(\partial^2A),hB),
       (h(\partial A),h(\partial B)),(hA,h(\partial^2B))\}\)
\ForAll{\((P,Q)\in\mathcal Q\)}
    \State Run
    \(\Call{LiftedTrial}{P,Q,2\varphi_0T}\) twice independently
    \If{both trials fail}
        \State \Return \(\mathsf{fail}\)
    \EndIf
    \State Let \(F\) be either successful ordinary convolution and
       store its reduction \(\overline F\), where, for
       \(r\in\mathbb Z_m\),
       \[
           \overline F_r
           :=\sum_{\substack{0\le q\le2m-2\\q\equiv r\pmod m}}F_q
           =F_r+F_{r+m},
       \]
       interpreting an out-of-range coefficient as zero
\EndFor
\State Form \(X,Y,Z\) from the six stored cyclic convolutions
\State \(\widetilde C\gets0\)
\For{\(r\in\supp(X)\)}
    \If{\(Y_r^2=X_rZ_r\)}
        \State \(z\gets Y_r/X_r\)
        \If{\(z\) is an integer in \(\{0,\ldots,2n-2\}\)}
            \State
            \(\widetilde C_z\gets\widetilde C_z+X_r\)
        \EndIf
    \EndIf
\EndFor
\If{\(\norm{\widetilde C}_1=\norm{A}_1\norm{B}_1\)}
    \State \Return \(\widetilde C\)
\Else
    \State \Return \(\mathsf{fail}\)
\EndIf
\end{algorithmic}
\end{algorithm}

\begin{lemma}[Convolution trial]
\label{lem:convolution_trial}
\Cref{alg:convolution_trial} takes
\(O(T\log(T+2))\) operations.  It never returns an
incorrect answer.  If \(T\ge t\), it returns \(A\star B\) with
probability at least \(9/10\).
\end{lemma}

\begin{proof}
We first prove soundness.  Every successful inner call returns its 
convolution by \cref{lem:lifted_trial}.  The variance test accepts only
buckets to which at most one coordinate contributes.  Since the buckets
partition the input contributions, every coefficient added to
\(\widetilde C\) is part of the corresponding coefficient of
\(A\star B\), and hence
\(\widetilde C\le A\star B\) coordinatewise.  Moreover,
\(\norm{A\star B}_1=\norm{A}_1\norm{B}_1\).  The final mass test
therefore succeeds only when \(\widetilde C=A\star B\).

Now suppose \(T\ge t\).
By \cref{lem:reduced_sparsity}, each of the six ordinary convolutions
has at most \(2|\Phi|t\le2\varphi_0T\) nonzero coordinates.  Each call
to \(\textsc{LiftedTrial}\) therefore succeeds with probability greater
than \(9/10\).  The probability that both trials fail for any fixed
convolution is less than \(1/100\), so all six convolutions are obtained
with probability at least \(1-6/100\).

Call \(h\) good if
\[
    h(z)+\phi\ne h(z')+\phi'
\]
for all distinct \(z,z'\in\supp(A\star B)\) and all
\(\phi,\phi'\in\Phi\).  For fixed \(z,z',\phi,\phi'\), equality would
imply \(h(z)-h(z')=\phi'-\phi\) in \(\mathbb Z_m\).  The
uniform-difference guarantee and a union bound give
\[
    \Pr[h\text{ is not good}]
    \le\frac{\kappa|\Phi|^2t^2}{m}
    \le\frac1{100},
\]
where the final inequality uses \(|\Phi|\le\varphi_0\) and \(t\le ab\).
If \(h\) is good, we show that each bucket contributes only to one output coordinate,
so the test succeeds.
Assume $x + y \ne z + w$, but $h(x) + h(y) = h(z) + h(w)$, then let $h(x) + h(y) = h(x + y) + \psi_0$ and $h(z) + h(w) = h(z + w) + \psi_1$, then $h(x + y) + \psi_0 = h(z + w) + \psi_1$, contradicting goodness. Thus the algorithm succeeds
with probability at least \(1-6/100-1/100>9/10\) and recovers every coordinate.  In the case of $2n \le m$,
the same guarantee follows from \cref{lem:lifted_trial}.

It remains to bound the runtime.  Put \(s=\max\{a,b,2\}\).  Since
\(m=O((ab)^2)=O(s^4)\) and \(T\ge s\), every invocation of
\(\textsc{LiftedTrial}\) has length \(T^{O(1)}\).  Its two input
supports have size at most \(s\), and its sparsity guess is \(O(T)\).
A constant number of such invocations therefore takes
\(O(T\log(T+2))\) time.

The hashed inputs are constructed sparsely.  For each nonzero input
coordinate, we generate a record consisting of its hash value and the
appropriate coefficient, discard zero records, sort the
records by hash value, and sum records with equal hash values.  Over
all six hashed vectors this takes
\(O(T\log(T+2))\) time.

Every successful inner call returns \(O(T)\) nonzero records.  Folding
their indices modulo \(m\), coalescing equal residues, and forming
\(X,Y,Z\) can likewise be done by sorting and merging
\(O(T)\) records in \(O(T\log(T+2))\) time.  The vector
\(\widetilde C\) is maintained as a sparse dictionary, so the update
\(\widetilde C_z\gets\widetilde C_z+X_r\) aggregates contributions from
different buckets having the same recovered coordinate.  Sampling and
evaluating \(h\) takes constant time per input record, and scanning
\(\supp(X)\) takes \(O(T)\) time.  All of this is within the claimed
bound.
\end{proof}

\mainthm*

\begin{proof}
By the zero-input convention in the preliminaries, we may assume that
\(t\ge1\).  Let
\[
    s=\max\{a,b,2\},
    \qquad U=2n.
\]
Run \cref{alg:convolution_trial} independently at the distinct scales
\[
    T_0=s,\qquad T_{j+1}=\min\{2T_j,U\}.
\]
At every scale below \(U\), move to the next scale if the trial fails.
Once \(U\) is reached, repeat independent trials at the fixed scale
\(T=U\) until one succeeds.  Every returned result is correct by
\cref{lem:convolution_trial}, and \(U\ge t\) because
\(t\le2n-1\).

Let \(T^\star\) be the first scheduled scale satisfying
\(T^\star\ge t\).  By \cref{lem:input_support},
\(s=O(t)\).  Since successive distinct scales grow by at most a factor
of two and the cap satisfies \(U\ge t\), it follows that
\(T^\star=O(t)\).  The total work at scales below \(T^\star\) is
\[
    O(T^\star\log(T^\star+2)).
\]

Set
\[
    R:=\frac{U}{T^\star}\ge1,
    \qquad q:=\lceil\log_2R\rceil,
    \qquad \gamma:=\log_2 10>1.
\]
Before the capped scale, the probability of reaching the \(j\)-th
scale after \(T^\star\) is at most \(10^{-j}\).  Hence the expected
work at the uncapped scales from \(T^\star\) onward is at most
\[
    \sum_{j=0}^{q-1}
       10^{-j}\,
       O\!\left(
          2^jT^\star\log(2^jT^\star+2)
       \right)
    =O(T^\star\log(T^\star+2)).
\]
Reaching the capped scale requires \(q\) consecutive failed trials and
therefore has probability at most
\[
    10^{-q}\le R^{-\gamma}.
\]
Conditional on reaching the cap, the expected number of trials at
\(U\) is at most \(10/9\).  Its unconditional expected contribution is
therefore
\[
\begin{aligned}
    O\!\left(10^{-q}U\log(U+2)\right)
    &=
    O\!\left(
       T^\star R^{1-\gamma}
       \bigl(\log(T^\star+2)+\log R\bigr)
    \right)\\
    &=O(T^\star\log(T^\star+2)),
\end{aligned}
\]
where the last equality uses \(\gamma>1\).  This proves the expected
running-time bound.  Since every scale is at most \(2n\), every array
length, address, and sampled modulus used in this execution is
\(n^{O(1)}\), as required by the fixed-word RAM model.

For the tail bound, let
\[
    q_\delta:=\left\lceil\log_{10}(1/\delta)\right\rceil.
\]
Run \(q_\delta\) independent trials at each distinct scale before
advancing; at the capped scale \(U\), continue with independent trials
if the first \(q_\delta\) all fail.  The probability that all
\(q_\delta\) trials at \(T^\star\) fail is at most
\(10^{-q_\delta}\le\delta\).  On the complementary event, the algorithm
terminates by the end of the trials at \(T^\star\), after
\[
    O\!\left(
       T^\star\log(T^\star+2)\,q_\delta
    \right)
    =
    O\!\left(
       t\log(t+2)\log(1/\delta)
    \right)
\]
work.  On the remaining event the algorithm continues at scales no
larger than \(U\) and still returns only a certified result, preserving
the Las Vegas guarantee.
\end{proof}

\paragraph{Acknowledgements}
The author thanks Jason Li for suggestions regarding the exposition of this paper.

\bibliographystyle{alpha}
\bibliography{references}

\end{document}